\documentclass[conference,10pt]{IEEEtran}
\usepackage{times}
\usepackage{flushend}
\usepackage{cite}
\usepackage{color}
\usepackage{epsf}
\usepackage{epsfig}
\usepackage{graphicx}
\usepackage{graphics}
\usepackage{epstopdf}
\usepackage[footnotesize]{caption}
\usepackage{amsmath}
\usepackage{amssymb}
\usepackage{amsthm}
\usepackage{amsxtra}
\usepackage{algorithm}
\usepackage{algorithmic}
\usepackage{enumerate}
\usepackage{multirow}
\usepackage{enumerate}
\usepackage{amssymb}
\usepackage{lipsum}
\usepackage{setspace}
\usepackage{multirow}

\newtheorem{theorem}{\bf Theorem}
\newtheorem{proposition}{\bf Proposition}

\usepackage{rotating} 
\usepackage{graphicx} 

\newcommand{\Rmnum}[1]{\expandafter\@slowromancap\romannumeral #1@}

\makeatother
\begin{document}
\title{Feasibility of Using Discriminate Pricing Schemes for Energy Trading in Smart Grid}
\author{\IEEEauthorblockN{Wayes Tushar\IEEEauthorrefmark{1}, Chau Yuen\IEEEauthorrefmark{1}, Bo Chai\IEEEauthorrefmark{2}, David B. Smith\IEEEauthorrefmark{3}, and  H. Vincent Poor\IEEEauthorrefmark{4}}
\IEEEauthorblockA{\IEEEauthorrefmark{1}Singapore University of Technology and Design, Singapore 138682. Email: \{wayes\_tushar, yuenchau\}@sutd.edu.sg. \\\IEEEauthorrefmark{2}State Key Lab. of Industrial Control Technology, Zhejiang University, China. Email: chaibozju@gmail.com.\\
\IEEEauthorrefmark{3}NICTA, Canberra, ACT, Australia. Email: david.smith@nicta.com.au.\\
\IEEEauthorrefmark{4}School of Engineering and Applied Science, Princeton University, Princeton, NJ, USA. Email: poor@princeton.edu.}
\thanks{This work is supported by the Singapore University of Technology and Design (SUTD) under the Energy Innovation Research Program (EIRP) Singapore NRF2012EWT-EIRP002-045.}
\thanks{\IEEEauthorrefmark{3}David B. Smith is also with the Australian National University (ANU), and his work is supported by NICTA. NICTA is funded by the Australian Government through the Department of Communications and the Australian Research Council through the ICT Centre of Excellence Program.}
}
\IEEEoverridecommandlockouts
\maketitle
\begin{abstract}
This paper investigates the feasibility of using a \emph{discriminate} pricing scheme to offset the inconvenience that is experienced by an energy user (EU) in trading its energy with an energy controller in smart grid. The main objective is to encourage EUs with small distributed energy resources (DERs), or with high sensitivity to their inconvenience, to take part in the energy trading via providing incentive to them with relatively higher payment at the same time as reducing the total cost to the energy controller. The proposed scheme is modeled through a two-stage Stackelberg game that describes the energy trading between a shared facility authority (SFA) and EUs in a smart community. A suitable cost function is proposed for the SFA to leverage the generation of discriminate pricing according to the inconvenience experienced by each EU. It is shown that the game has a unique sub-game perfect equilibrium (SPE), under the certain condition at which the SFA's total cost is minimized, and that each EU receives its best utility according to its associated inconvenience for the given price. A backward induction technique is used to derive a closed form expression for the price function at SPE, and thus the dependency of price on an EU's different decision parameters is explained for the studied system. Numerical examples are provided to show the beneficial properties of the proposed scheme.
\end{abstract}
 \setcounter{page}{1}
\section{Introduction}\label{sec:introduction}
Energy management (or demand management) is a technique that changes the electricity usage patterns of end users in response to the changes in the price of electricity over time~\cite{Albadi-JEPSR:2008,Liu-STSP:2014}. With the advancement of distributed energy resources (DERs), the technique can also be used to assist the grid or other energy controllers such as a shared facility authority (SFA)~\cite{Tushar-ISGT:2014} to operate reliably and proficiently by supplying energy to them~\cite{Tushar-TSG:2014}. %

The majority of energy management literature focuses mainly on three different pricing schemes: time-of-use pricing; day-ahead pricing; and real-time pricing~\cite{Yi-TSG:2013}. Time-of-use pricing~\cite{Asano-TPS:1992} has three different pricing rates: peak, off-peak and shoulder rate based on the use of electricity at different times of the day. Day-ahead pricing~\cite{Torre-TPS:2002} is, in principle, determined by matching offers from generators to bids from energy users (EUs) so as to develop a classic supply and demand equilibrium price at an hourly interval. Finally, real-time pricing~\cite{Yi-TSG:2013} refers to tariffed retail charges for delivering electric power and energy that vary hour-to-hour, and are determined from wholesale market prices using an approved methodology. Other popular dynamic pricing schemes include critical peak pricing, extreme day pricing, and extreme day critical peak pricing~\cite{Yi-TSG:2013}. It is important to note that in all of the above mentioned pricing schemes all EUs are charged at the same rate at any particular time.

Due to government subsidies to encourage the use of renewables~\cite{Fischer-JEEM:2008}, more EUs with DERs are expected to be available in smart grid. This will lead to a better completion of a purchasing target for an energy controller, and thus more saving from its buying cost. Particularly, for energy controllers such as an SFA that relies on the main grid as its primary source of energy~\cite{Tushar-ISGT:2014}, the opportunity for trading energy with EUs can greatly reduce their dependency, and consequently decrease their cost of energy purchase~\cite{Tham-JTSMCS:2013}. Nevertheless, not all EUs would be interested in trading energy with the energy controller if the benefit is not attractive~\cite{Naveed-Energies:2013}. This can precisely happen to EUs with merely limited energy capacity, or to EUs that are highly sensitive to the inconvenience caused by the trading of energy whose expected return could be very small. In this case, the EUs would store the energy or change its consumption schedule rather than selling it to the energy controller~\cite{Tushar-TIE:2014}. However, one possible way to address this is to pay them a relatively higher price per unit of energy, compared to the EUs with very large DERs, without affecting their revenue significantly. In fact, allowing discriminate pricing not only considerably benefits EUs with lower energy capacity without significantly affecting others, as we will see shortly, but also benefits the SFA by reducing its total cost of energy purchase when adopting this flexible pricing.
\begin{table}
\caption {Numerical example of a discriminate pricing scheme where an SFA requires $40$ kWh of energy from two EUs and the SFA's total price per unit of energy to pay to the EUs is $40$ cents/kWh.}
\centering
\begin{tabular}{|c||c|c|}
\hline
& \textbf{Case 1} & \textbf{Case 2}\\
\hline
Payment to EU1 (cents/kWh) & $20$ & $18$\\
\hline
Payment to EU2 (cents/kWh)& $20$ & $22$\\
\hline
Energy supplied by EU1 (kWh) & $35$ & $32$\\
\hline
Energy supplied by EU2 (kWh) & $5$ & $8$\\
\hline
Revenue of EU 1 (cents) & $700$ & $576$ (-$17\%$)\\
\hline
Revenue of EU 2 (cents) & $100$ & $176$ (+$76\%$)\\
\hline\hline
\textbf{Cost to the SFA} (cents) & $800$ & $752$ (-$6\%$)\\
\hline
\end{tabular}
\label{table:motivation}
\end{table}

For instance, consider the numerical example given in Table \ref{table:motivation} where the SFA buys its required $40$ kWh energy from EU1 and EU2. EU1 has $50$ kWh and EU2 has $10$ kWh of energy to sell to the SFA. In case 1, the SFA pays the same price $20$ cents/kWh to each of them, and EU1 and EU2 sell $35$ and $5$ kWh respectively to the SFA. Hence, the revenues of EU1 and EU2 are $700$ and $100$ cents respectively, and the total cost to the SFA is $800$ cents. In case 2, the SFA uses discriminate pricing to motivate EU2 to sell more to the SFA. Therefore, it pays $22$ cents/kWh to the EU2 and $18$ cents/kWh to EU1. Now, due to this increment of price EU2 increases its selling amount to $8$ kWh, and the SFA procures the remaining $32$ kWh from EU1. Therefore, the revenues changes to $576$ and $176$ cents for EU1 and EU2 respectively, and total cost to the SFA reduces to $752$ cents. Thus, from this particular example it can be argued that discriminate pricing can be considerably beneficial to EUs with small energy (revenue increment is $76\%$) in expense of relatively lower revenue degradation (e.g., $17\%$ in the case of EU1) from EUs with larger DERs. It also reduces the cost to the SFA by $6\%$. Therefore, discriminate pricing is advantageous for reducing SFA's cost and also for circumstances where the SFA motivates the participation of EUs with both large and small DERs in the energy trading. Hence, there is a need for investigation as to how this pricing scheme can be adopted in a smart grid environment.

To this end, we take \emph{the first step} towards discussing the properties of a discriminate pricing scheme. The idea of discriminate pricing was first used to design a consumer-centric energy management scheme in \cite{Tushar-TSG:2014}. However, no insight was provided into the choice of different prices that are paid to different EUs. In this paper, we first propose a scheme by using a two-stage Stackelberg game. In the proposed scheme, the EUs with smaller energy generation can expect higher unit selling price, and the price is adaptive to their available energy for sale and their sensitivity to the inconvenience of energy exchange. At the same time, the scheme is designed to minimize the total purchasing cost to the energy controller whereas each EU also receives its best utility based on its available energy, its sensitivity to the inconvenience, and the offered price by the SFA. We prove the existence of a solution to the proposed game, and use a backward induction method to determine how the unit price set by the energy controller is affected by an EU's various parameters. We further derive a closed form expression for differing price generation considering some conditions on the energy controller's cost function. Finally, we present some numerical cases to show the properties of the  proposed discriminate pricing scheme.

We stress that current grid systems do not allow such discriminate pricing among EUs. However, we envision it as a further addition to real-time pricing schemes in future smart grid. Examples of such differentiation can also be found in standard Feed-in-Tariff (FIT) schemes~\cite{2012-solarchoice}.

\section{System Description and Problem Formulation}\label{sec:system-model}
\begin{figure}[b!]
\centering
\includegraphics[width=\columnwidth]{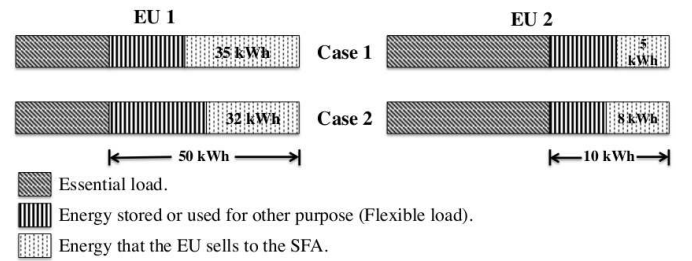}
\caption{Example of how EUs are sensitive to the inconvenience caused by a change of price, which thus affects their amount of energy to trade with the SFA.} \label{fig:example}
\end{figure}
Consider a smart community consisting of a large number of EUs, an SFA and the main electric grid. Each EU can be a single user, or group of users connected via an aggregator that acts as a single entity~\cite{Fang-J-CST:2012}.  EUs are equipped with distributed energy resources (DERs) such as wind turbines and solar arrays. They can sell their energy, if there is any remaining after meeting their essential loads, to the SFA or to the main grid to make some extra revenue. Since the grid's buying price is significantly low in general~\cite{McKenna-JIET:2013}, it is reasonable to assume that each EU would be more interested in selling its energy to the SFA instead of selling to the grid. Alternatively, an EU can store its energy or schedule its equipment instead of selling the energy to the SFA if the return benefit is not attractive, i.e., if the price is not convenient enough for the EU to trade its energy. We briefly explain this phenomenon  by an example in Fig.~\ref{fig:example}.

In Fig.~\ref{fig:example}, we use the same example of Table~\ref{table:motivation} and show how sensitive an EU is to the inconvenience of trading energy caused by the change of price per unit of energy. As can be seen from the figure, EU1 has considerably lower essential load than EU2, and thus has a larger available energy to supply to the SFA. As a result, in case 1, EU1 supplies $35$ kWh of energy to the SFA whereas EU2 supplies $5$ kWh for the same per unit price of $20$ cents/kWh after using the energy for their other flexible loads. However, in case 2, the SFA adopts a discriminate pricing scheme and changes the per unit price to $18$ cents/kWh and $22$ cents/kWh to pay to EU1 and EU2 respectively. Due to the change of price, the expected return for EU2 becomes larger from trading its energy at the expense of the revenue degradation from EU1. Consequently, energy trading becomes more inconvenient for EU1 where as at the same time it becomes more appealing for EU2. As shown in Fig.~\ref{fig:example}, due to their sensitivities to the inconvenience caused by the change of price, EU1 reduces its amount of energy for selling to $32$ kWh (i.e., by increasing its use of the remaining available energy for other purposes such as storage) whereas EU2 increases its amount of energy for selling to $8$ kWh in case 2. In this paper, we quantify this sensitivity of each EU to the relative inconvenience through an \emph{inconvenience parameter\footnote{Where a higher and lower value of this parameter refers to the higher and lower sensitivity of an EU respectively to the inconvenience caused by energy trading.}}, as we will see shortly, and  analyze its effects on the total cost to the SFA. The SFA refers to an energy controller\footnote{For the rest of this paper, we will use SFA to indicate an energy controller as discussed in Section~\ref{sec:introduction}.} that controls the electricity consumed by the equipment and machines that are shared and used by EUs on daily basis. The SFA does not have any energy generation capacity, and therefore depends on EUs and the main grid for its required energy. The SFA is connected to the main grid and all EUs via power and communication lines~\cite{Fang-J-CST:2012}.

To this end, let us assume that $N$ EUs in a set $\mathcal{N}$ are taking part in energy trading with the SFA. At a particular time of the day, the SFA's energy requirement is $E_r$, and each EU $i\in\mathcal{N}$ has an available energy of $E_i$ after meeting its essential load from which it can sell $e_i$ to the SFA. The main objective of each EU $i$ is to make some extra revenue by selling $e_i$ to the SFA at a price $c_i$ per unit of energy. However, the choice of $e_i$ is reasonably affected by the inconvenience parameter $\alpha_i$, which is a measure of sensitivity of EU $i$ to the inconvenience it faces to trade its energy. In this regard, we define a utility function $U_i$ for each EU $i$ that captures the effect of this inconvenience, and is assumed to possess the following properties:
\begin{enumerate}[i)]
\item The utility function is an increasing function of $e_i$ and $c_i$, and a decreasing function of inconvenience parameter $\alpha_i$. That is $\frac{\delta U_i}{\delta e_i}, \frac{\delta U_i}{\delta c_i}>0$, and $\frac{\delta U_i}{\delta \alpha_i}<0$. $\alpha_i$ captures the fact that the utility will decrease for an EU if its sensitivity to the inconvenience of trading energy increases.
\item The utility function is a concave function of $e_i$, i.e., $\frac{\delta^2U_i}{\delta {e_i}^2}<0$. Therefore, the utility can become saturated or even decrease with an excessive $e_i$. This can be interpreted by the fact that since EUs with DERs are equipped with a battery with limited capacity in general, excessive supply of energy once exceeding a certain limit would risk the depletion of battery due to the aging effect upon the battery, and consequently decrease the EU's utility.
\end{enumerate}
Formally, we define $U_i~\forall i$ as
\begin{eqnarray}
U_i = e_ic_i + (E_i-\alpha_i e_i)e_i.
\label{eqn:utility}
\end{eqnarray}
In \eqref{eqn:utility}, $e_ic_i$ is the direct income that the EU $i$ receives from selling its energy to the SFA at a price $c_i$ per unit of energy. $(E_i-\alpha_i e_i)e_i$ refers to the possible loss  for the EU's inconvenience sensitivity $\alpha_i>0$. Different values of $\alpha_i$ reflect different negative impacts of energy supply on an EU's utility, and an EU can set higher $\alpha_i$ if it prefers to sell less. For example, the effects of $c_i$ and $\alpha_i$ on an EU's utility from its energy trading is shown in Fig.~\ref{fig:effect-utility}.  Now, with the goal of maximizing utility, the objective of each EU can be expressed as
\begin{eqnarray}
\max_{e_i}\left[e_ic_i + (E_i - \alpha_ie_i)e_i\right].\label{eqn:obj-eu}
\end{eqnarray}
\begin{figure}[t!]
\centering
\includegraphics[width=\columnwidth]{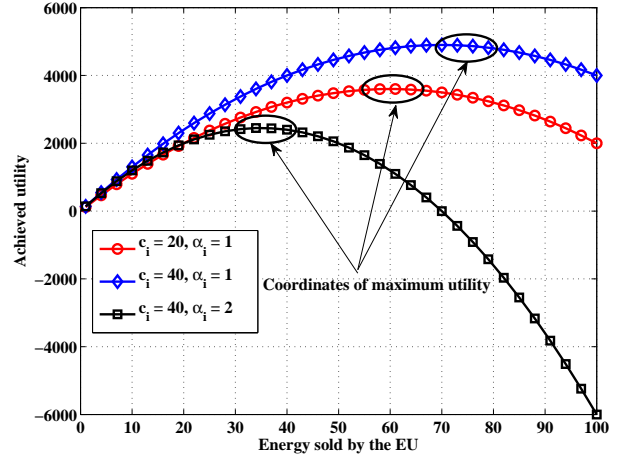}
\caption{Effect of parameters such as $\alpha_i$ and $c_i$ on the achieved utility of an EU are shown in this figure. As can be seen in the figure, for the same $\alpha_i$ a higher $c_i$ encourages an EU to sell more to the SFA and thus the maximum utility shifts towards a higher value on the right.  By contrast, a higher $\alpha_i$ causes more inconvenience to the EU which can lead the utility even to a negative value (i.e., cost) for greater energy trading.} \label{fig:effect-utility}
\end{figure}

However, as a buyer of energy, the SFA wants to minimize its total cost $J$ of energy purchase from EUs and the grid. In this paper, we consider the following cost function to capture the total cost to the SFC for buying its required energy from EUs and the grid:
\begin{eqnarray}
J = \sum_{i=1}^N \left(e_ic_i^k + a_ic_i + b_i\right)+c_g\left(E_r - \sum_i e_i\right),\label{eqn:cost}
\end{eqnarray}
such that
\begin{eqnarray}
\sum_i c_i \leq C,~c_\text{min}\leq c_i\leq c_\text{max},\label{eqn:const-sfa}
\end{eqnarray}
where $C$ is the total unit energy price~\cite{Tushar-TSG:2014}, and $c_\text{min}$ and $c_\text{max}$ are the lower and upper limits of unit price that the SFA can pay to any EU~\cite{Tushar-TSG:2014}. In \eqref{eqn:cost}, $e_ic_i^k$ corresponds to the direct cost $e_ic_i$ that is weighted by $c_i^{k-1}$ to generate discriminate prices for EUs with different $\alpha_i$, and the term $(a_ic_i + b_i),~a_i,b_i>0$ accounts for other costs such as transmission cost and store of purchased energy cost~\cite{Tushar-TSG:2014}. $c_g\left(E_r - \sum_i e_i\right)$ is the cost of purchasing energy from the grid. $C$ scales a set of normalized prices to generate the unit price $c_i$. It is fixed for a particular time and can be determined by the SFA using any real-time price estimator, e.g., the estimator proposed in~\cite{2008IEEE-JTPS_Yun}.

Now, the SFA's objective is to set a price $c_i$ per unit of energy for each EU $i$ that not only minimizes its total cost in \eqref{eqn:cost} but also pays a price to each EU according to their inconvenience parameters, and thus encourages them to take part in energy trading with the SFA. Therefore, the objective of the SFA can be defined as
\begin{eqnarray}
\min_{c_i}\left[\sum_{i=1}^N \left(e_ic_i^k + a_ic_i + b_i\right)+c_g\left(E_r - \sum_i e_i\right)\right],\label{eqn:obj-sfa}
\end{eqnarray}
such that \eqref{eqn:const-sfa} is satisfied.

We stress that \eqref{eqn:obj-eu} and \eqref{eqn:obj-sfa} are related via $e_i$ and $c_i$, and can be solved in a centralized fashion. However, considering that the nodes in the system are distributed, it is more advantageous to define a solution approach that can be implemented distributedly according to the parameter setting within the system~\cite{Rad-JTSG:2010}. In this regard, we propose to use a game theoretic formulation. In \cite{Tushar-TSG:2014}, the effect of changing $C$ on the cost to a seller was investigated, and a distributed algorithm was proposed to design a consumer-centric smart grid via capturing this effect. In this paper, we focus on exploring the influence of different EUs' behavior on the choice of price by the SFA, and the resultant cost incurred to it. To that end, we propose a two-stage Stackelberg game in the next section.
\section{Two-Stage Stackelberg Game}\label{sec:stackelberg-game}
To determine energy trading parameters $e_i$ and $c_i$, on the one hand, each EU $i$ needs to decide on the amount of energy $e_i$ that it wants to sell to the SFA according to its inconvenience sensitivity and the offered price. On the other hand, based on the amount of energy offered by each EU and its inconvenience parameter, the SFA agrees on the price vector $\mathbf{c} = \left[c_1, c_2, \hdots, c_N\right]$ that it wants to pay to each EU such that the cost $J$ to the SFA is minimized. Thereupon, this sequential interaction can be modeled as a two-stage Stackelberg game~\cite{Bu-TEPC:2013}, which is formally defined as
\begin{eqnarray}
\Omega = \{(\mathcal{N}\cup\{\text{SFA}\}), \{\mathbf{E}_i\}_{i\in\mathcal{N}}, \{U_i\}_{i\in\mathcal{N}}, J, \mathbf{c}\}.\label{eqn:game-formal}
\end{eqnarray}
In \eqref{eqn:game-formal}, $(\mathcal{N}\cup\{\text{SFA}\})$ is the set of total players in the game where each EU $i\in\mathcal{N}$ is a follower, and $\{\text{SFA}\}$ is the leader. $\mathbf{E}_i$ is the strategy vector of each follower $i$ and $U_i$ is the utility that the follower $i$ receives from choosing its strategy $e_i\in\mathbf{E}_i$. $J$ is the cost incurred to the SFA for choosing the strategy vector $\mathbf{c}$.

As the leader of $\Omega$, the SFA chooses its strategy vector $\mathbf{c}$ in the first stage of the game such that its cost function in \eqref{eqn:cost} is minimized, and the constraints in \eqref{eqn:const-sfa} are satisfied. In the second stage of the game, each EU $i\in\mathcal{N}$ independently chooses $e_i$ in order to maximize its utility in \eqref{eqn:utility} in response to $c_i$ chosen by the SFA. Consequently, $\Omega$ reaches the equilibrium solution of the game.
\subsection{Solution Concept}
A general solution of a multi-stage Stackelberg game such as the proposed $\Omega$ is the sub-game perfect equilibrium (SPE)~\cite{Bu-TEPC:2013}. A common method to determine the SPE of a Stackelberg game is to adopt a backward induction technique that captures the sequential dependencies of decisions between stages of the game~\cite{Bu-TEPC:2013}. To that end, we first analyze how each EU would maximize its benefit by playing its best response to the price offered by the SFA in stage two. Then, we explore how the SFA decides on different prices to pay to different EUs according to their offered energy and inconveniences. We note that, due to the method for game formulation, $\Omega$ will possess a SPE if there exists a solution in both stages of the decision making process by the SFA and EUs. In fact, the existence of a solution in pure strategies is not always guaranteed in a game~\cite{ChaiBo-TSG:2014}, and hence there is a need to investigate the existence of a solution in the proposed $\Omega$.
\begin{theorem}\label{thm:theorem-1}
A unique SPE exists for the proposed two-stage Stackelberg game $\Omega$ if $k=2$ in \eqref{eqn:cost}.
\end{theorem}
\begin{proof}
According to the backward induction technique, each EU $i\in\mathcal{N}$ decides on their energy trading parameters $e_i~\forall i$ at the second stage of the game to maximize \eqref{eqn:obj-eu}. It is a strictly concave function of $e_i$ as $\frac{\delta^2U_i}{\delta c_i^2} = -2\alpha_i$ and $\alpha_i>0$. Hence, EU's decision making problem has a unique solution. Furthermore, in the first stage of the game, the SFA optimizes its price $c_i$ to pay to each EU $i$. Now, we note that if $k=2$ in \eqref{eqn:cost}, which is a general choice of quadratic cost function for electricity utility companies and controllers~\cite{Tushar-TSG:2014,Rad-JTSG:2010},  the cost function \eqref{eqn:cost} is strictly convex with respect to $c_i$. Thus, for the amount of energy offered by each EU, the choice of different price to pay to each $i$ also possesses a unique solution, which minimizes \eqref{eqn:obj-sfa}. Hence, the game $\Omega$ possesses a unique SPE, and thus Theorem~\ref{thm:theorem-1} is proved.
\end{proof}
\subsection{Analysis of Energy Trading Behavior}
In this section, we show how the energy trading behavior of the SFA and EUs are affected by different decision making parameters such as the price set by the SFA, and the inconvenience that is caused to each EU for trading its energy. First, we consider the second stage of the game where each EU $i$ plays its best response to the price $c_i$ offered by the SFA. Since the utility function in \eqref{eqn:utility} is differentiable, we obtain the first order derivative $\frac{\delta U_i}{\delta e_i}$, and $U_i$ attains its maximum when $\frac{\delta U_i}{\delta e_i}=0$. Therefore, from \eqref{eqn:utility}, the best response function of EU $i$ to a given $c_i$ can be expressed as
\begin{eqnarray}
e_i^*(c_i) = \frac{c_i + E_i}{2\alpha_i},\label{eqn:relation-energy-price}
\end{eqnarray}
which leads to the following proposition:
\begin{proposition}
For an offered price $c_i$, the amount of energy $e_i$ that an EU $i$ is willing to sell, from its available energy $E_i$, to the SFA decreases with the increase of its sensitivity to inconvenience $\alpha_i$. In other words, an EU with inconvenience parameter $\alpha_i$ would be more willing to sell its energy to the SFA for a higher price per unit of energy.
\end{proposition}

The SFA's cost, on the other hand, is determined by the price $c_i^*$ that it wants to pay to each EU $i$ for its offered energy $e_i^*$. Therefore, in the first stage of the game the SFA determines the price $c_i^*~\forall i$ having the knowledge of the energy vector $\mathbf{e} = [e_1^*, e_2^*, \hdots, e_N^*]$ offered by all EUs via \eqref{eqn:relation-energy-price}. Now, the Lagrangian for the SFA's optimization problem in \eqref{eqn:obj-sfa} is given by
\begin{eqnarray}
\Gamma = \sum_i \left(e_i^*c_i^{k}+a_ic_i+b_i\right)&+&c_g(E_r-\sum_i e_i^*)\nonumber\\ &+&\lambda(C-\sum_i c_i),
\label{eqn:lagrange-approach}
\end{eqnarray}
where $\lambda$ is the Lagrange multiplier and
\begin{eqnarray}
\frac{\delta\Gamma}{\delta c_i} = 0.\label{eqn:lagrange-diff}
\end{eqnarray}
 In \eqref{eqn:lagrange-approach}, we only consider the case when $c_\text{min}\leq c_i\leq c_\text{max}$, i.e., the Lagrange multiplier associated with $c_\text{min}$ and $c_\text{max}$ are assumed to be zero\footnote{The conditions for $c_i^* = c_\text{min}$ and $c_i^* = c_\text{max}$ are considered at the solution of the $c_i$ in \eqref{eqn:price-final}.}. Now replacing the value $e_i^*$ in \eqref{eqn:lagrange-approach} from \eqref{eqn:relation-energy-price}, \eqref{eqn:lagrange-diff} can be expressed as
 \begin{eqnarray}
 \frac{k+1}{2\alpha_i}c_i^k + \frac{kE_i}{2\alpha_i}c_i^{k-1} + a_i - \frac{c_g}{c_n}-\lambda = 0.
 \label{eqn:lagrange-diff-2}
 \end{eqnarray}
 Now, for the general case\footnote{We will consider $k=2$ for the rest of the paper.} $k = 2$,
 \begin{eqnarray}
 3c_i^2 + 2E_ic_i + 2\alpha_i(a_i-\lambda) - c_g = 0,
 \label{eqn:lagrange-diff-3}
 \end{eqnarray}
 and consequently,
 \begin{eqnarray}
 c_i = \frac{-E_i + \left[E_i^2 - 3\left(2\alpha_i(a_i-\lambda)-c_g\right)\right]^{\frac{1}{2}}}{3}.\label{eqn:price-closed-form}
 \end{eqnarray}
In \eqref{eqn:price-closed-form}, $\lambda$ and $a_i~\forall i$ are design parameters, and thus constant for a particular system. $\lambda$ needs to be chosen  significantly higher than $a_i~\forall i$ such that $c_i$ always possesses a positive value. Note that we skip the other solution of $c_i$ in \eqref{eqn:price-closed-form} for the same reason.

From \eqref{eqn:price-closed-form}, we note that for the same generation and grid price, \emph{a higher price needs to be paid to an EU $i$ with higher inconvenience parameter $\alpha_i$ compared to an EU $j,~j\not= i$ with $\alpha_j<\alpha_i$ to encourage it to sell energy}. Nevertheless, in cases when $c_i>c_\text{max}$ and $c_i<c_\text{min}$, the SFA sets $c_i$ to the respective limits. Hence, the choice of price by the SFA to pay to each EU $i$ at the SPE can be expressed as
\begin{eqnarray}
c_i^* = \begin{cases}
c_\text{min}, & \text{$c_i<c_\text{min}$}\\
\frac{-E_i + \left[E_i^2 - 3\left(2\alpha_i(a_i-\lambda)-c_g\right)\right]^{\frac{1}{2}}}{3}, & \text{$c_\text{min}\leq c_i\leq c_\text{max}$}\\
c_\text{max}, & \text{$c_i>c_\text{max}$}
\end{cases}.
\label{eqn:price-final}
\end{eqnarray}
\section{Case Study}\label{sec:numerical-experiment}
\begin{figure}[b!]
\centering
\includegraphics[width=\columnwidth]{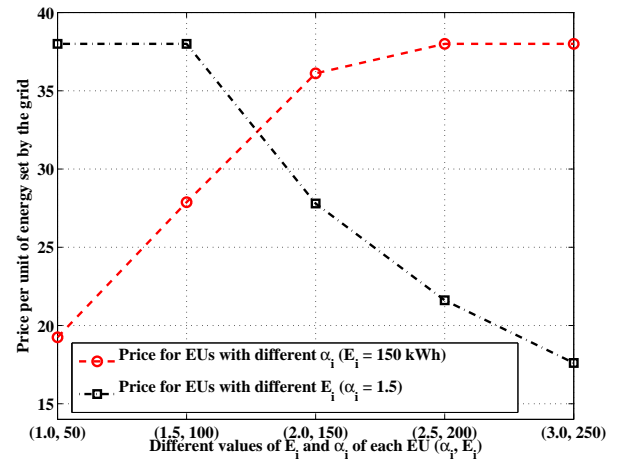}
\caption{Effect of available energy and inconvenience parameters on the price per unit of energy that the SFA selects to pay to each EU. The first term of each tuple on the horizontal axis refers to the inconvenience parameter $\alpha_i$ of an EU, and the second term indicates the available energy $E_i$.} \label{fig:figure-5}
\end{figure}

To show the properties of the proposed discriminate pricing scheme, we consider an example in which a number of EUs are interested in trading their energy with the SFA in the time slot of interest. We assume that the available energy to each EU, after meeting its essential load, is uniformly distributed within [50,~250], and the energy required by the SFA for the considered time slot is $650$ kWh. The value of $\lambda$ is chosen to be $1000$. The grid's selling price is set to be $50$ cents/kWh, and $c_\text{max}$ and $c_\text{min}$ are assumed to be $38$ and $10$ cents/kWh\footnote{Price $c_\text{min}$ is marginally greater than the price of $8.45$ cents/kWh that a grid typically pays to buy energy from DERs~\cite{2012-solarchoice}.} respectively. These two values are chosen such that the SFA can pay to each EU a price, which is lower than the grid's selling price, and at the same time is higher than the grid's buying price. This condition is necessary to motivate all EUs to trade their energy only with the SFA instead of the grid. Nonetheless, we highlight that all parameter values are chosen particularly for this case study only and that these values may vary between different case studies.

In Fig.~\ref{fig:figure-5}, we show how the price per unit of energy is decided by the SFA for each EU. According to \eqref{eqn:price-final}, for a particular grid price $c_g$, the unit price $c_i$ that the SFA pays to each EU $i$ depends on 1) EU's inconvenience parameter $\alpha_i$, and 2) the available energy $E_i$ to each EU. First, we consider five EUs with the same $E_i = 150$ kWh, but with different inconvenience in selling their energy to the SFA. We note that the SFA tends to pay more, within the constraint in \eqref{eqn:const-sfa}, to the EU with higher sensitivity to inconvenience. In fact, a higher inconvenience parameter refers to the state at which trading energy with the SFA is not a convenient option for an EU. Therefore, to encourage the EU to sell the energy the SFA needs to increase its unit price to pay. However, if $c_i$ becomes more than $c_\text{max}$, the SFA pays $c_\text{max}$ to the EU as shown in the case of the last EU with inconvenience parameter $\alpha_i=3$ in Fig.~\ref{fig:figure-5}.

By contrast, for the same sensitivity to inconvenience, the SFA pays a higher price to an EU with lower available energy and vice versa. In fact, a lower available energy could stop an EU from selling the energy to the SFA as it might not bring significant benefit to the EU at a lower price. Hence, to provide more incentive to the EU, the SFA needs to pay a relatively higher price per unit of energy. However, EUs with larger amount of energy can still obtain higher utilities from trading a considerable amount of energy with the SFA even at a relatively lower price, as explained by the example in Table \ref{table:motivation}. Thus, the SFA pays comparatively a lower price to such EUs to minimize the cost of energy trading, such that the energy trading does not effect their utilities significantly\footnote{We note that  the lowest price per unit of energy $c_\text{min}$ is assumed to be higher than the buying price of the grid. Therefore, any EU with a higher available energy would benefit more from trading with the SFA instead of trading with the grid.}.

After showing how prices are set by the SFA to pay to different EUs, we now show the effect of different behavior of EUs in a group on the total cost to the SFA. First we note that EUs' behaviors are dominated by their inconvenience parameters $\alpha$. For example, if an EU with very large available energy does not want to sell its energy at the offered price it can set its $\alpha$ high and thus insignificantly (or, not at all) take part in energy trading. Hence, we can model different EUs behaviors by simply changing their $\alpha_i~\forall i$.  To this end, we assume a network with $10$ EUs that have the same available energy $150$ kWh but different inconvenience parameters to sell their energy to the SFA. For this particular case, we consider total unit energy price $C = 380$ cents/kWh so that even when all EUs are paid at $c_\text{max}$, the constraints in \eqref{eqn:const-sfa} are still satisfied, and the unit price for each EU remains lower than the grid's selling price. We compare the performance with an equal distribution scheme (EDS) such as in \cite{Wayes-J-TSG:2012}, where $C$ is equally divided to pay to each EU for buying its energy. That is in EDS each EU is paid a price\footnote{For $N=10$, each EU is paid a price $38$ cents/kWh.} $\frac{C}{N}$ per unit of energy, where $N$ is the total EUs in the network.

\begin{table}[h!]
\centering
\caption{Different behavioral cases of EUs in the network (a total of 10 EUs) -- where the number of EUs with a particular inconvenience parameter $\alpha_i \in \{1,2,3\}$ is specified.}
\begin{tabular}{|c|c|c|c|}
\hline
Cases & $\alpha_i = 1$ & $\alpha_i = 2$ & $\alpha_i = 3$\\
\hline
1 & 6 EUs & 2 EUs & 2 EUs\\
\hline
2 & 4 EUs & 3 EUs & 3 EUs\\
\hline
3 & 2 EUs & 4 EUs & 4 EUs\\
\hline
4 & 2 EUs & 2 EUs & 6 EUs\\
\hline
5 & 1 EU & 1 EU & 8 EUs\\
\hline
6 & 0 & 0 & 10 EUs \\
\hline
\end{tabular}
\label{table:1}
\end{table}

\begin{figure}[t!]
\centering
\includegraphics[width=\columnwidth]{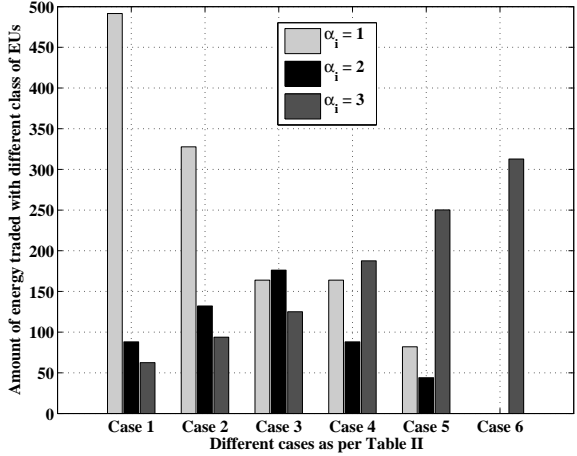}
\caption{Effect of behavior of EUs on the total amount of energy that the SFA buys from each class of EUs.} \label{fig:figure-6}
\end{figure}

To that end, we categorize the behavior of EUs into six different cases based on the number of EUs with particular inconvenience parameters $\alpha_i$ in the group as shown in Table~\ref{table:1}. Although we have chosen only three integer values of $\alpha_i\in\{1, 2, 3\}$, other fractional values within this range are equally applicable to define different levels of sensitivity to inconvenience. Now first we see from Fig.~\ref{fig:figure-6} that as the number of EUs with $\alpha_i = 1$ dominates the group, the SFA buys most of its energy from them. For example, in case 1 and case 2, the number of EUs with $\alpha_i=1$ is higher in the system and consequently, the SFA buys significantly large amount of energy from them in these two cases compared to the other cases, as shown in Fig~\ref{fig:figure-6}. However, as their number reduces the SFA needs to buy more energy from the other two types of EUs, based on their percentage of presence in the group, with relatively higher payment. In the extreme case, i.e., case 6, the SFA needs to buy all its energy from EUs with $\alpha_i = 3$ as there are no other types of EUs in the system. Consequently, this trend of energy trading affects the total cost to the SFA from buying its energy from EUs and the grid. We show these effects separately in Table \ref{table:2}.
\begin{table*}[t!]
\centering
\caption{Cost to the SFA in dollars for different EUs' behaviors (cases stated in Table \ref{table:1}).}
\begin{tabular}{|c|c|c|c|c|c|c|}
\hline
Different Costs & Case 1 & Case 2 & Case 3 & Case 4 & Case 5 & Case 6\\
\hline
Cost for buying from EUs with $\alpha_i = 1$ & 68.92 & 45.94 & 22.97 & 22.97 & 11.48 & 0\\
\hline
Cost for buying from EUs with $\alpha_i = 2$ & 23.6 & 35.4 & 47.2 & 23.6 & 11.8 & 0 \\
\hline
Cost for buying from EUs with $\alpha_i = 3$ & 24.36 & 38.4 & 48.72 & 73.08 & 97.44 & 121.8\\
\hline
Cost for buying from the grid & 53.91 & 98.16 & 142.42 & 155.18 & 186.88 & 218.58\\
\hline
Total cost for proposed scheme & 170.79 & 216.06 & 261.32 & 274.89 & 307.62 & 341\\
\hline
Total cost for EDS & 341 & 341 & 341 & 341 & 341 & 341\\
\hline
$\%$ reduction in total cost & $49.91\%$ & $36.63\%$ & $23.36\%$ & $19.38\%$ & $9.78\%$ & $0\%$\\
\hline
\end{tabular}
\label{table:2}
\end{table*}

From Table \ref{table:2}, first we note that the amount of energy that the SFA buys from the grid increases as the categories of EUs change from case 1 to case 6 in the system. This is due to the fact that as the number of EUs with higher sensitivity to inconvenience increases in the group, the total amount of power that the SFA can trade with EUs becomes lower. Hence, the SFA needs to procure the remainder of required energy from the grid at a higher price. Secondly, the cost to the SFA to buy energy from EUs with higher inconvenience parameters also increases its cost significantly as the SFA needs to pay a higher price to them. For example, consider the different cost that is incurred to the SFA for buying energy from different types of EUs in case $1$. From Fig~\ref{fig:figure-6}, we can see that the amount of energy that the SFA buys from EUs with $\alpha_i = 1$ is almost five times the amount it buys from EUs with $\alpha_i = 2$ and $3$. However, the resultant cost is only three times more than the cost to buy from EUs with higher sensitivity. Therefore, more EUs with lower sensitivity to inconvenience allows the SFA to procure more energy at a comparatively lower cost. Therefore, the total cost incurred by the SFA increases significantly with an increase in the number of EUs with higher inconvenience parameters as can be seen from Table \ref{table:2}.

We also compare the total cost that is incurred to the SFA with the case when the SFA adopts an EDS scheme for energy trading in Table \ref{table:2}. In an EDS scheme, the cost to the SFA remains the same for all type of EU groups as the cost does not depend on their categories. From Table \ref{table:2}, the proposed scheme shows considerable benefit for the SFA in terms of reduction in total cost when there are a relatively higher number of EUs with lower inconvenience parameters in the group. For example, as shown in Table~\ref{table:2}, the cost reduction for the SFA is $49.9\%$ and $36.63\%$ respectively for case 1 and 2.  According to the current case study, the average total cost reduction for the SFA is $23.18\%$ compared to the EDS. However, the cost increases with the increase of number of EUs with high inconvenience parameters, and becomes the same as the EDS scheme when all the EUs in the group become highly sensitive to the inconvenience of energy trading, i.e., $\alpha_i = 3,~\forall i$ as can be seen from Table \ref{table:2}.

\section{Conclusion}\label{sec:conclusion}
In this paper, a discriminate pricing scheme has been studied to counterbalance the inconvenience experienced by energy users  (EUs) with distributed energy resources (DERs) in trading their energy with other entities in smart grids. A suitable cost function has been designed for a shared facility authority (SFA) that can effectively generate different prices per unit of energy to pay to each participating EU according to an inconvenience parameter for the EU. A two-stage Stackelberg game, which has been shown to have a unique sub-game perfect equilibrium, has been proposed to capture the energy trading between the SFA and different EUs. The properties of the scheme have been studied at the equilibrium by using a backward induction technique. A theoretical price function has been derived for the SFA to decide on the price that it wants to pay to each EU, and the properties of the scheme are explained via numerical case studies. By comparing with an equal distribution scheme (EDS), it has been shown that discriminate pricing gives considerable benefit to the SFA in terms of reduction in total cost. One interesting future extension of the proposed scheme would be to design an algorithm that can capture the decision making process of the SFA and EUs in a distributed fashion. Also, finding a mathematical theorem that would explain the benefits to the SFA due to the discriminate pricing scheme is another possible extension of this work. Finally, the design of a scheme (i.e., game) with imperfect information about the inconvenience parameters also warrants future investigation.


\end{document}